\documentclass[12pt]{amsart}

\usepackage{amsfonts,latexsym,amsmath,amsthm,amssymb, color}
\usepackage[margin=1.0in]{geometry}

\def\V{{\mathcal V}}

\def\H{{\mathcal H}}

\def\C{{\mathbb C}}

\def\R{{\mathbb R}}

\newcommand{\diag}{\mathsf{diag} }

\newcommand{\Ran}{\mathsf{Ran}~ }

\newcommand{\scal}[1]{\langle#1\rangle}

\newtheorem*{theorem*}{Theorem}
\numberwithin{equation}{section}

\newtheorem{theorem}{Theorem}[section]
\newtheorem{proposition}[theorem]{Proposition}
\newtheorem{lemma}[theorem]{Lemma}
\newtheorem{corollary}[theorem]{Corollary}
\newtheorem{definition}  [theorem] {Definition}

\begin{document}

\title[Perturbation theory for the spectral...]{Perturbation theory for the spectral decomposition of Hermitian matrices}
\author{Marcus Carlsson}
\address{Centre for Mathematical Sciences, Lund University\\Box 118, SE-22100, Lund,  Sweden\\}
\email{marcus.carlsson@math.lu.se}

\begin{abstract}
Let $A$ and $E$ be Hermitian self-adjoint matrices, where $A$ is fixed and $E$ a small perturbation. We study how the eigenvalues and eigenvectors of $A+E$ depend on $E$, with the aim of obtaining first order formulas (and when possible also second order) that are explicitly computable in terms of the spectral decomposition of $A$ and the entries in $E$. In particular we provide explicit Fr\'{e}chet type differentiability results. The findings can  be seen as an extension of the Rayleigh-Schr\"{o}dinger coefficients for analytic expansions of one-dimensional perturbations.
\end{abstract}
\maketitle

%\begin{keyword}
%\MSC[2010] 15B05 \sep 15A03  \sep 47B35
%\end{keyword}

\section{Introduction}

Given a Hermitian self-adjoint matrix $A$, the spectral decomposition gives a unitary matrix $U_A$ and a diagonal matrix $\Lambda_\alpha$ such that
\begin{equation}\label{spec}
A=U_A\Lambda_\alpha U_A^*.
\end{equation}
The columns of $U_A$ are the eigenvectors whereas the elements on the diagonal of $\Lambda_\alpha$ are the corresponding eigenvalues of $A$, and we denote the vector of these by $\alpha$. This is one of the most fundamental results of matrix theory and its generalization to Hilbert spaces is a key tool in mathematical analysis (see e.g.~Ch.~XI,XIII and XIV of \cite{dunford1963linear}), numerical analysis \cite{higham2008functions} as well as mathematical physics and quantum mechanics \cite{hilbert1955methods,kemble1937fundamental,reed1980functional}. It is therefore rather surprising that the perturbation theory of this result has not yet been fully understood, and it is the aim of this article to improve the situation.

More precisely, given a ``small'' self-adjoint matrix $E$ we ask how $U_A$ and $\alpha$ change upon replacing $A$ with $A+E$. The literature on this topic is immense, specially concerning perturbation of the eigenvalues, and can roughly be divided into two groups. One group ``freeze'' the variable $E$ and consider $A+tE$ as a function of the complex variable $t$, giving rise to a beautiful and rich connection with algebra and complex function theory. However, it lacks a global perspective, in the sense that $E$ is fixed and not a free variable. The second group of results do not ``freeze'' $E$, with weaker more general results as a consequence, and is closer to the vein of this study.

A complicating factor is that $U_A$ is not unique. As long as the eigenvalues of $A$ are simple, this is a minor issue (the eigenvectors are then unique up to multiplication by unimodular numbers) but the higher multiplicity case is more intricate. For simplicity, let us first assume that all eigenvalues are simple and present our findings in this case. 

\subsection{Simple eigenvalues and Fr\'{e}chet differentiability}\label{secsimple}

First of all it is crucial to work in the basis given by the columns of $U_A$, so that $A$ simply reduces to the diagonal matrix $\Lambda_\alpha$. Concretely this can be done by noting that \begin{equation*}\label{y5}U_A^*(A+E)U_A=\Lambda_\alpha+U_A^*E U_A=\Lambda_\alpha+\hat E\end{equation*} where $\hat E$ is defined as $U_A^* E U_A$. Due to the assumption that $A$ has simple eigenvalues, $\hat E$ is independent of the choice of $U_A$. Suppose we are interested in understanding the leading order perturbation of the $j:$th eigenvalue $\alpha_j$. As a representative of known ``global'' perturbation theory results, let us mention that the theory of Ger\v{s}gorin discs states that, given an index $j$, there is an eigenvalue $\xi_j$ of $\Lambda_\alpha+\hat E$ which lies in a disc with center $\alpha_j+\hat E_{(j,j)}$ and radius given by the sum of the off-diagonal elements of the $j$:th row. Since $\Lambda_\alpha$ is diagonal this gives us the estimate \begin{equation}\label{gersgorin}|\xi_j-(\alpha_j+\hat E_{(j,j)})|\leq \sum_{i\neq j}|\hat E_{(j,i)}|.\footnote{The reason for the ugly parenthesis around the subindex is that we shall introduce another meaning for e.g.~$E_{11}$ in Section \ref{gre}.}\end{equation}
For the present purposes, this is not satisfactory since the ``perturbation'' $\hat E_{(j,j)}$ and the error estimate are both of order $O(\|E\|)$. We will in this paper show that this estimate can be improved to yield \begin{equation}\label{gersgorinalacarlsson}\xi_j=\alpha_j+\hat E_{(j,j)}+O(\|E\|^2)\end{equation} where $\xi_j$ is the $j:$th eigenvalue of $A+E$ ordered non-increasingly,
showing that the leading order perturbation of the $j:$th eigenvalue is indeed given by $\hat E_{(j,j)}$. As a consequence, the vector of eigenvalues $\xi$ (ordered non-increasingly) is Fr\'{e}chet differentiable at 0 (as a function of $E$). More precisely, letting $\H_n$ denote the set of $n\times n$ complex Hermitian matrices and $\xi'$ the Fr\'{e}chet derivative of $\xi$, we have;
\begin{theorem}\label{t1}
The map $\xi:\H_n\rightarrow\R^n$ is Fr\'{e}chet differentiable at 0 with Fr\'{e}chet derivative
\begin{equation}\label{vdf}\xi'(E)=(\hat E_{(11)},\ldots, \hat E_{(n,n)})=(u_1^*Eu_1,\ldots, u_n^*Eu_n),\end{equation}
where $u_1,\ldots,u_n$ are the eigenvectors of $A$ (the columns of $U_A$).
\end{theorem}
The formula \eqref{vdf} is a special case of Theorem 1.1 in \cite{lewis1996derivatives} and appears in a more general form in Corollary 10 of \cite{lewis1999nonsmooth}, we refer to that paper for further remarks on its history. In either case, it is not clear if \eqref{gersgorinalacarlsson} is implied by these results, and even \eqref{vdf} is hard to find in the standard literature. If we replace $E$ by $tE$ and freeze $E$, then it is well known that $\hat E_{(j,j)}$ is the first order term in the series expansion, which is shown in many books on mathematical physics or quantum mechanics, see e.g.~Sec.~XII.1 in \cite{reed1980functional}, but this also does not prove the stronger result \eqref{gersgorinalacarlsson}. In either case, the proofs provided here are completely different than previous approaches and the main goal of this article is to take the estimate \eqref{gersgorinalacarlsson} one step further and get $O(\|E\|^3)$ control.

Before discussing this issue, let us now focus on Fr\'{e}chet differentiability of the eigenvectors. We would like to consider the orthonormal matrix in the spectral decomposition as a function of $E$, but since this matrix is not unique this is not well defined. In Section \ref{secspec} we present a deterministic way of picking a concrete function $U=U(E)$ whose columns are the unit norm eigenvectors of $A+E$, such that this function is Fr\'{e}chet differentiable with an explicit formula. To state this, we introduce the matrix $M=M(A)$ by setting $M_{(j,j)}=0$ for $j=1,\ldots,n$ and \begin{equation}\label{defM1}M_{(i,j)}=\frac{1}{\alpha_i-\alpha_j},\quad i\neq j,\end{equation} and let $\circ$ denote Hadamard multiplication of matrices. Also, let $O_n$ denote the orthogonal group, i.e.~the set of all unitary matrices. Our result then reads as follows.
\begin{theorem}\label{t2}
Given a fixed matrix $U_A$ in the spectral decomposition of $A$, there exists a map $U:\H_n\rightarrow O_n$ whose columns are eigenvectors to $A+E$ such that $U(0)=U_A$, which is Fr\'{e}chet differentiable at 0 with Fr\'{e}chet derivative
\begin{equation}\label{Frechet eig}U'(E)= -U_A(M\circ \hat E).\end{equation}
\end{theorem}

This result is shown in a more general context in Section \ref{secspec}. More precisely, we remove the assumption that the eigenvalues of $A$ are simple, but then we have to restrict formula \eqref{Frechet eig} to a subset of $E$'s.

\subsection{$O(\|E\|^3)$ control on the eigenvalues}\label{secO3}

Neither Theorem \ref{t1} nor \ref{t2} are very hard to prove, the main purpose of the paper is to provide a more precise control on the eigenvalues of the perturbation. This may seem like a technicality but nevertheless an important one. For example, it is crucial when we want to extend Theorem \ref{t2} to the case when $A$ does not have simple eigenvalues. Moreover, in a sequel article \cite{carlsson2018perturbation2} we shall prove a simple new formula for $\sqrt{A+E}$ which rather surprisingly seems to be new (see Sec.~\ref{secapp}), and again the key ingredient in the proof is the results which we now present. Unfortunately, it is not true that \eqref{gersgorinalacarlsson} holds with $O(\|E\|^2)$ replaced by $O(\|E\|^3)$, but rather one needs to introduce certain Schur complements to gain this extra precision.

As before $\alpha$ denotes the eigenvalues of $A$ and $\xi$ those of $A+E$, both ordered non-increasingly. Pick a particular eigenvalue $\alpha_j=\rho$ of $A$ with multiplicity $l\geq 1$. In case $l>1$ then let $\alpha_{i+1}$ be the first occurrence of $\rho$ and $\alpha_{i+l}$ the last. Denote the orthogonal projection onto the eigenspace defined by $\rho$ by $P_\rho$, and set $P_{\rho}^\perp=I-P_\rho$. Set $C=P_\rho E P_{\rho}^\perp$, $D=P_\rho^{\perp} E P_{\rho}^\perp$ and $$B=P_\rho E P_{\rho}-C\Big(P_\rho^{\perp} (A-\rho I+E) P_{\rho}^\perp\Big)^\dagger C^*$$ where $\dagger$ indicates the Moore-Penrose inverse. $B$ is a sort of Schur complement of the submatrix $P_\rho^{\perp} (A-\rho I+E) P_{\rho}^\perp$ in the matrix $A-\rho I+E$. Note that $B$ has rank less than or equal to $l$ and let $\beta=(\beta_1,\ldots,\beta_l)$ denote the $l$ first eigenvalues of $B$ ordered non-increasingly.  We then have

\begin{theorem}\label{teigenvalues}
With $\alpha,~\beta,~\xi$ as above, we have
\begin{equation}\label{est2}\xi_j=\alpha_j+\beta_{j-i}+O(\|B\|\|C\|^2),\quad i<j\leq i+l.\end{equation}
\end{theorem}
Clearly $O(\|B\|\|C\|^2)\leq O(\|E\|^3)$, so the estimate \eqref{est2} is a bit sharper than promised. It is for example interesting to observe that the matrix $D$ is absent in the error estimate, which of course does not imply that it does not affect $\xi_j$, it simply states that its effect on $\xi_j$ is very small.

As a consequence of Theorem \ref{teigenvalues}, it is easy to see that the estimate \eqref{gersgorinalacarlsson} still holds, given that we have chosen the eigenvectors (the columns of $U_A$) so that whenever $\alpha_i=\alpha_j$, we have $\hat E_{(i,j)} =0$ and \begin{equation}\label{gr}\hat E_{(i,i)}\geq \hat E_{(j,j)}\end{equation} for $j>i$ (we prove this in Sec.~\ref{secO2}). To pave the way for the next section, we give such matrices a special name. \begin{definition}\label{defblock} Matrices $\hat E$ that satisfy the above requirements will be called ``{block-wise diagonal}'', letting the dependence on the ``blocks'' of $A$ (i.e.~sets of indices such that $\alpha_i=\alpha_j$) be implicit. If the inequality in \eqref{gr} is strict for all possible pairs $(i,j)$, then we say that $\hat E$ has ``{block-wise decreasing diagonal elements}''.\end{definition} Although $\hat E$ in general depends on the particular choice of $U_A$, we remark that $\hat E$ is unique when $\hat E$ is block-wise diagonal with block-wise decreasing diagonal elements, since then the columns of $U_A$ are unique up to multiplication with unimodular numbers.

\subsection{Gateaux type differentiability of the eigenvalues/vectors}
Concerning perturbation of the eigenvectors in this more general framework, it is not possible to provide a nice global result as Theorem \ref{t2}, and the situation is more intricate. Indeed, it is well known that it is impossible to define $U(E)$ in a way such that the eigenvectors are continuous in a neighborhood of $0$. Despite this, when considering perturbations along a line, i.e.~$A+tF$ where $t\in\R$ and $F$ is fixed, it is well known that the eigenvectors can be defined as analytic functions, a most surprising result due to F. Rellich. We will give a concrete construction of such a function $U(t)$, i.e.~so that its columns are normalized eigenvectors of $A+tF$, under the assumption that $\hat F$ is block-wise diagonal with block-wise decreasing diagonal elements. More interestingly, we give an explicit formula for its derivative at 0. However, it comes as a surprise that the formula \eqref{Frechet eig} does not apply without further modification. An explicit example of this is given in Section \ref{secspec}. The details how \eqref{Frechet eig} needs to be modified are rather involved so we omit them from this introduction and refer to Section \ref{secline} and Theorem \ref{t5}.

Concerning the eigenvalues $\xi=\xi(t)$ to $A+tF$ the situation is better. We recall the celebrated result by F. Rellich which states that these become real analytic functions of $t$ (Chapter 3.5, \cite{baumgartel1985analytic}), (at the price of removing the convention that $\xi(t)$ is non-increasing for every fixed $t$). In either case, this result is remarkable in the light of the erratic behavior of the eigenvectors near $A$. Explicit expressions for the coefficients in the corresponding series expansion is known in the mathematical physics community as the Rayleigh-Schr\"{o}dinger-coefficients, and in the case when $A$ has only simple eigenvalues their form is known up until arbitrary order (see e.g.~Sec. XII.1 of \cite{reed1980functional}). However, it is hard to find information about the case when $A$ has eigenvalues of higher multiplicity, although it appears e.g.~in the classic \cite{hilbert1955methods} by Courant and Hilbert. In fact it has been rediscovered for example in \cite{lancaster1964eigenvalues}. As an easy consequence of the framework developed in this paper we retrieve these formulas
\begin{equation}\label{tao}\xi_j(t)=\alpha_j+t\hat F_{(j,j)}+t^2\sum_{k:\alpha_k\neq \alpha_j}\frac{|\hat F_{(k,j)}|^2}{\alpha_j-\alpha_k}+O(t^3).\end{equation}
This is shown in Theorem \ref{t4}. For some reason this extension has not been picked up by mainstream books relating to this topic, so we here wish to highlight the result and provide an alternative proof. A corresponding expression for the eigenvalues is found in Theorem \ref{t5} of Section \ref{secline}, but again this expression is found already in \cite{hilbert1955methods}.

\subsection{Application; expansions of $\sqrt{A+E}$ and $|A+E|$}\label{secapp}

Another consequence of Theorem \ref{teigenvalues}, which we will develop in a separate publication \cite{carlsson2018perturbation2}, is a new perturbation theory for functional calculus of matrices, also known as matrix functions. As an example, suppose that both $A$ and $E$ are positive semi-definite, and suppose we are interested in approximating $\sqrt{A+E}$ for small $E$. If we introduce $B,C$ and $D$ for $\rho=0$, as in Section \ref{secO3}, we will show that $\sqrt{A+ E}$ can be written as $\sqrt{A}$ plus a term that depends linearly on $C, D$ and $\sqrt B$, where the error is $O(\|E\|^{3/2})$. Hence, the function $\sqrt{A+E}$ is sort of ``Fr\'{e}chet differentiable'' except for the non-linearity in the contribution $E\mapsto\sqrt{B}$, which nevertheless is given by a concrete construction. Based on this we also provide similar perturbation results for the matrix modulus $E\mapsto|A+E|$ for arbitrary (not necessarily self-adjoint) matrices $A,E$. The study of the matrix square root began with Cayley in 1858 \cite{cayley1858ii}, and in the light of its tremendous influence on pure and numerical analysis (see e.g.~Ch.~6-8, \cite{higham2008functions}), it is rather surprising that its perturbation theory is not yet fully understood.

\subsection{Related works}\label{secrel}

Matrix perturbation theory is an old and well-studied subject, of interest to engineers, physicists, applied and pure mathematicians. We list a few important contributions and relate key results to the theory developed here. It seems that E. Schr\"{o}dinger was one of the first to postulate some results and conjectures \cite{schrodinger1926quantisierung}, and these results are of key interest to mathematical physics and in particular quantum physics, where more complicated systems are considered as perturbations of simpler systems for which closed form solutions do exist. The classical example is the study of the hydrogen atom, see Example 3, Section XII.2, \cite{reed1980functional}. Many more interesting examples are found in the same chapter, and the ``Notes''-section contains a more extensive historical exposition.

F. Rellich was the first to systematically study the topic in a sequence of papers in the 30's and 40's (St\"{o}rungstheorie der Spektralzerlegung I-V), and in particular he proved Schr\"{o}dinger's conjectures and established analyticity of eigenvalues and eigenprojections for perturbations (depending analytically on one parameter) of self-adjoint operators. The area was very active through the 50's and 60's which led to the classic \cite{kato2013perturbation} by T.~Kato (see in part.~Sec.~6, Ch II), still today a key reference on matrix perturbation theory. This work is continued e.g.~in \cite{baumgartel1985analytic} as well as here in Section \ref{secline}. The main difference is that previous results are based on deep analytic and algebraic function theory, and not direct constructions as performed in this paper. Consequently, previous results are stronger in the sense that they provide stunning facts such as that both eigenvalues and eigenvectors can be taken as analytic functions, but weaker in the sense that they force you to freeze a direction $F$ and study only perturbations of the form $tF$ for $t\in\R$. %For example, the proof that the line perturbation of a Hermitian matrix yield analytic eigenvalues is based on contradiction; the eigenvalues are first shown to be branches of analytic functions with at most algebraic singularities, then it is shown that if these are not trivial, the eigenprojections have yet worse singularities, but this contradicts the fact that eigenprojections of normal matrices are well behaved.
As mentioned earlier, the coefficients in \eqref{tao} are known as the Rayleigh-Schr\"{o}dinger coefficients.

In parallel, global bounds for the perturbation of eigenvalues goes back to H. Weyl around 1910 \cite{weyl1912asymptotische}, where in particular the famous ``Weyl perturbation theorem'' is established. Improvements were then given e.g.~by Hoffman-Wielandt, Bauer-Fike, Mirsky and later Bhatia. It seems that \eqref{gersgorinalacarlsson} is a sort of local improvement of Weyl's, Bauer-Fike's and Hoffman-Wielandt's theorems, in the sense that these results give less accurate information on the eigenvalues of $A+E$ than \eqref{gersgorinalacarlsson} for small $E$, but on the other hand they apply globally as opposed to the asymptotic estimates given here. We refer to \cite{horn1990matrix} (Ch. 6), \cite{stewart1990matrix} (Ch. IV and V) and \cite{bhatia2013matrix} (Ch. VI and VII) for more information on this type of results.

Other more recent contributions to perturbation theory for Hermitian matrices include \cite{axelsson2006eigenvalue,hiriart1995sensitivity,lewis1999nonsmooth,mathias1997spectral,sun2002strong}, but the results are of a different nature than those presented here. For example, Section 9 of \cite{lewis1999nonsmooth} tries to understand the local behavior of eigenvalues using so called Clarke subdifferentials. In particular, the formula \eqref{vdf} appears in Corollary 10 of \cite{lewis1999nonsmooth} in a more general context. The recent article \cite{sjostrand2007elementary} treats the use of Schur complements in spectral problems, but seems to have no overlap with the present article. See also Ch. 15 of \cite{hogben2006handbook} for an overview of modern results. The fairly recent book \cite{zhang2006schur} contains a compilation of known uses of Schur complements, see in particular Section 2.5 containing a large amount of estimates on eigenvalues in the Hermitian case.

From the numerical perspective we mention the books \cite{golub2012matrix,higham2008functions,parlett1998symmetric} and \cite{wilkinson1965algebraic}. It seems that the results provided in this paper could be of interest also from an applied perspective, for example for fast evaluation of the singular value decomposition of a matrix $A+E$, given that the one for $A$ is known and that $E$ is small. However, this has to be investigated elsewhere.

Concerning perturbations of eigenvectors we refer to Ch. VII of \cite{bhatia2013matrix}. One of the key theorems in this field is by Davis and Kahan, which relates the angle between a spectral subspace of an operator $A$ with a different spectral subspace of a perturbation $A+E$, (see \cite{davis1970rotation} or Theorem VII.3.4 of \cite{bhatia2013matrix}). However, we have not found any concrete results on how the eigenvectors behave, similar to Theorem \ref{t2} or its extensions Proposition \ref{p1}.

\section{$O(\|E\|^2)$ control on the eigenvalues}\label{secO2}
In this section we prove Theorem \ref{t1} as well as the estimate \eqref{gersgorinalacarlsson} under the assumption that $\hat E$ is block-wise diagonal. Both these can be deduced as consequences of the refined estimate \eqref{est2} in Theorem \ref{teigenvalues}, but the proof we give here is much shorter and more elegant so it would be a pity not to include it. We thus consider the perturbation $A+E$ and we select a matrix $U=U_{A}$ so that $A=U_A\Lambda_\alpha U_A^*$ is the spectral decomposition of $A$, and moreover such that $\hat E$ is block-wise diagonal (see Definition \ref{defblock} in Section \ref{secO3}). Set $\hat E^d=\hat E\circ I$ and $\hat E^o=\hat E-\hat E^d$, where $d$ stands for ``diagonal'' and $o$ stands for ``off-diagonal''.
   We extend the definition of the matrix $M$ in \eqref{defM1} to include the case when $A$ has higher multiplicity, as follows \begin{equation}\label{M}M_{(i,j)}=\begin{cases}&\frac{1}{\alpha_i-\alpha_j},\quad \alpha_i\neq \alpha_j\\& 0, \quad \textit{else.}\end{cases}\end{equation} %where $sgn(0)=0$.
and note that
\begin{equation}\label{id}M\circ(\Lambda_\alpha \hat E-\hat E\Lambda_\alpha)=\Lambda_\alpha (M\circ \hat E)-(M\circ \hat E)\Lambda_\alpha=\hat E^o.\end{equation}
The key observation is that in the basis given by the columns of $I+M\circ \hat E$, the matrix $\Lambda_\alpha+\hat E$ becomes $\Lambda_\alpha+\hat E^d+O(\|E\|^2).$ To see this, note that $(I+M\circ \hat E)^{-1}=I-M\circ\hat E+O(\|E\|^2)$ and therefore \begin{equation}\label{lars}\begin{aligned}&(I+M\circ \hat E)(\Lambda_\alpha+\hat E)(I+M\circ \hat E)^{-1}=\\&\Lambda_\alpha+\hat E+(M\circ\hat E)\Lambda_\alpha-\Lambda_\alpha(M\circ\hat E)+O(\|E\|^2)=\Lambda_\alpha+\hat E^d+O(\|E\|^2).\end{aligned}\end{equation}
In particular $A+E$ and the matrix at the end of the above calculation share eigenvalues. As outlined in Section \ref{secsimple}, Ger\v{s}gorin's theorem now immediately implies that for each $j$, there exists an eigenvalue $\xi_j$ to $A+E$ such that \eqref{gersgorinalacarlsson} holds (i.e.~$\xi_j=\alpha_j+\hat E_{(j,j)}+O(\|E\|^2)$). Since we have assumed that $\hat E_{(i,i)}\geq \hat E_{(j,j)}$ whenever $i<j$ and $\alpha_i=\alpha_j$, it is clear that the ordering of $\xi$ can be chosen as the usual non-increasing one for small enough $E$. The above also implies Theorem \ref{t1}, because when $A$ has only simple eigenvalues,  $E\mapsto\diag(\hat E)=(u_1^*Eu_1,\ldots, u_n^*Eu_n)$ is a well defined linear functional that is independent of the choice of $U_A$.

Although this proof is fairly elementary, it is rather surprising that it seems to have gone unnoticed. Indeed, as argued in Section \ref{secsimple}, Ger\v{s}gorin's theorem applied to the original matrix $\Lambda_\alpha+\hat E$ only gives an error bound $O(\|E\|)$, which is not sufficient to determine the leading order perturbation. Nevertheless, the improved estimate \eqref{gersgorinalacarlsson} is not mentioned in any of the standard books on the topic (see Sec.~\ref{secrel}), and we have not been able to locate it elsewhere either. Usually, one finds the theorem by Bauer-Fike which for normal matrices states that $|\xi_j-\alpha_j|\leq \|E\|$, but for Hermitian matrices this is actually weaker than Weyl's theorem which states that $\alpha_j+\varepsilon_n\leq \xi_j\leq \alpha_j+\varepsilon_1$ where $\varepsilon_1,\ldots,\varepsilon_n$ are the eigenvalues of $E$. However, in terms of estimates of $|\xi_j-\alpha_j|$ by $\|E\|$, the statement $|\xi_j-\alpha_j|\leq \|E\|$ is clearly the best possible. It would therefore be wrong to claim that \eqref{gersgorinalacarlsson} improves Bauer-Fike or Weyl's results, it just gives more accurate information (locally) by providing the leading order term of the perturbation.

We have not been able to find any argument which leads to the stronger estimate \eqref{est2} in Theorem \ref{teigenvalues} based on Ger\v{s}gorin's theorem, and instead we will base our argument on a theorem of H.~Cartan (son of E.~Cartan) which we introduce in the next section.

\section{On a theorem by Henri Cartan}

The theorem of H.~Cartan we allude to is not his most famous result, but a very nice one with a beautiful proof. Like the theorem of Ger\v{s}gorin, it gives a number of discs with certain properties, but instead of dealing with matrices it deals with polynomials. More precisely it gives a lower bound for the modulus of a given polynomial outside of the discs, see e.g. Chapter I.7, \cite{levin1964distribution} or the original article \cite{cartan1928systemes}. We will apply this result to the characteristic polynomial of $A+E$ in order to establish the estimate \eqref{est2}. We here present a corollary of Cartan's theorem which is tailormade for our present purposes, and therefore rather involved.

\begin{lemma}\label{l2}
Let $\Omega_1,\Omega_2,\Omega_3\subset \C$ be open and bounded such that $\bar\Omega_1\subset\Omega_2$ and $\bar\Omega_2\subset\Omega_3$. Let $\mu_1,\ldots,\mu_k$ be points in $\Omega_1$ and consider \begin{equation}\label{g6}\varphi(\zeta)\prod_{j=1}^k(\zeta-\mu_j)+\psi(\zeta)\end{equation} where $\varphi~,\psi$ are analytic in a neighborhood of $\bar \Omega_3$ and moreover $|\varphi(\zeta)|>c_1$ for all $\zeta\in\bar\Omega_3$, where $c_1>0$ is some constant. Then there exist constants $c_2,~c_3>0$, independent of $\{\mu_j\}_{j=1}^k$, such that for all analytic $\psi$ with $\sup_{\zeta\in\Omega_3}|\psi(\zeta)|<c_2$, \eqref{g6} has precisely $k$ zeroes $\nu_1,\ldots,\nu_k$ in $\Omega_2$ and \begin{equation}\label{g8}\|\mu-\nu\|_2\leq c_3\left(\sup_{\zeta\in\Omega_3}|\psi(\zeta)|\right)^{1/k}.\end{equation} Moreover, \eqref{g6} can be written as
\begin{equation}\label{g7}\tilde\varphi(\zeta)\prod_{j=1}^k(\zeta-\nu_j)\end{equation} where $|\tilde\varphi(\zeta)|>c_1/2$ for all $\zeta\in\bar\Omega_2.$
\end{lemma}
\begin{proof}
We write $\|\psi\|=\sup_{\zeta\in\Omega_3}|\psi(\zeta)|$ in this proof. By Cartan's theorem on the minimum modulus of polynomials, there exists a collection of circles, each containing at least one $\mu_k$, such that the sum of the radii is bounded by $2e\|\psi/\varphi\|^{1/k}$ and \begin{equation}\label{gtde}|\prod_{j=1}^k(\zeta-\mu_j)|>\|\psi/\varphi\|\end{equation} outside of these circles. Since $\|\psi/\varphi\|\leq c_2/c_1$, it follows that the circles are inside of $\Omega_2$ if $c_2$ is small enough. By a basic residue calculus we conclude that the amount of zeroes in each circle is the same for $\varphi\prod_{j=1}^k(\zeta-\mu_j)$ and the perturbation $\varphi\prod_{j=1}^k(\zeta-\mu_j)+\psi$ (integrate $\frac{(\varphi(\zeta)\prod_{j=1}^k(\zeta-\mu_j)+t\psi(\zeta))'}{\varphi(\zeta)\prod_{j=1}^k(\zeta-\mu_j)+t\psi(\zeta)}$ on the boundary of each circle, $0\leq t\leq 1$, noting that the denominator is non-zero by \eqref{gtde}). Ordering them such that $\mu_j$ and $\nu_j$ is in the same circle, we conclude that \begin{equation}\label{po0}|\mu_j-\nu_j|\leq 4e\|\psi/\varphi\|^{1/k}\leq \frac{4e}{c_1^{1/k}}\|\psi\|^{1/k},\end{equation} from which \eqref{g8} follows.

In a similar manner we can use Rouch\'{e}'s theorem on $\Omega_3$ to show that the amount of zeroes is constant in $\Omega_3$, given that $c_2$ is sufficiently small. The function \begin{equation}\label{po1}\tilde\varphi(\zeta)=\frac{\varphi(\zeta)\prod_{j=1}^k(\zeta-\mu_j)+\psi(\zeta)}{\prod_{j=1}^k(\zeta-\nu_j)}\end{equation} is thus analytic and zero free on $\Omega_3$, so the minimum modulus theorem shows that it attains its minimum on the boundary. If $d$ is a separating distance between $\Omega_2$ and $\Omega_3^c$, we have $$\left|\frac{\zeta-\mu_j}{\zeta-\nu_j}\right|\geq 1-\frac{|\nu_j-\mu_j|}{|\zeta-\nu_j|}\geq 1-\frac{4ec_2^{1/k}}{c_1^{1/k}d}$$ for all $\zeta\in\partial\Omega_3$, where we used \eqref{po0}. The expression in \eqref{po1} is thus bounded below on $\Omega_3$ by $$\left(1-\frac{4e c_2^{1/k}}{c_1^{1/k}d}\right)^kc_1-\frac{c_2}{d^k}.$$ If $c_2$ is small enough, it follows that the inequality $\|\tilde\varphi\|>c_1/2$ can be achieved independently of $\{\mu_j\}_{j=1}^k$.

\end{proof}

\textit{Remark:} In the above lemma, the ordering of $\mu$ and $\nu$ becomes important. However, when applied to characteristic polynomials of Hermitian matrices, we know that $\mu$ and $\nu$ consist of real numbers, which have a natural non-increasing ordering. In this case, the inequality \eqref{g8} is valid with this ordering, as follows by a classical reordering lemma stating that, for fixed values in $\mu$ and $\nu$ but not fixed ordering, $\scal{\mu,\nu}$ is maximized by choosing both sequences non-increasingly ordered, see e.g.~Corollary II.4.3 in \cite{bhatia2013matrix}.

\section{$O(\|E\|^3)$ control on the eigenvalues}\label{gre}

We now prove Theorem \ref{teigenvalues}, i.e.~the estimate \eqref{est2}. Pick a particular eigenvalue $\rho$ of $A\in\H_n$ and denote its multiplicity by $l\geq 1$. Set $m=n-l$ and let $\tau$ be a vector containing the $m$ remaining eigenvalues. To explain the main ideas behind the proof, we first introduce a decomposition of the space in which $B,C$ and $D$ gets more concrete definitions, i.e.~we wish to omit the projections $P_\rho$ and $P_\rho^{\perp}$ from our formulas. To this end, we pick  a basis of orthonormal eigenvectors $U_A=[U_{A1}~ U_{A2}]$ where $U_{A1}$ spans the eigenspace of $\rho$ and $U_{A2}$ its orthogonal complement. We write the spectral decomposition \eqref{spec} as \begin{equation}\label{sang}A=U_A\Lambda_\alpha U_A^*=[U_{A1} ~U_{A2}]\left(
                                                                 \begin{array}{cc}
                                                                   \rho I_l & 0 \\
                                                                   0 & \Lambda_\tau \\
                                                                 \end{array}
                                                               \right) [U_{A1} ~U_{A2}]^*
\end{equation}
where $\Lambda_\tau$ is diagonal of size $m\times m$ and $I_l$ denotes the identity matrix of size $l\times l$. $\tau$ is thus the vector of eigenvalues not equal to $\rho$. Note that this implies an unusual ordering on the eigenvalues on the diagonal of $\Lambda_\alpha$, unless $\rho$ happens to be the largest one. Consider again a perturbation $A+E$ and let $\hat E$ as before equal $U_A^*E U_A$, which then naturally decomposes as \begin{equation}\label{u7}\hat E=\left(
                                                                 \begin{array}{cc}
                                                                   \hat E_{11} & \hat E_{12} \\
                                                                   \hat E_{21} & \hat E_{22} \\
                                                                 \end{array}
                                                               \right)=[U_{A1}~ U_{A2}]^*E [U_{A1}~ U_{A2}].\end{equation}
With these definitions we now introduce $ \hat E_{12}=C$, $\hat E_{22}=D$ and \begin{equation}\label{defB}B=\hat E_{11}-C(\Lambda_\tau-\rho I_m+D)^{-1}C^*.\end{equation} It follows by basic linear algebra that $B,~C,~D$ in the new basis are identical with $B,~C,~D$ as introduced in Section \ref{secO3}, except for technicalities. For example, $B$ (as introduced here) equals the matrix representation of the restriction of $B$ (as introduced there) to $\Ran P_\rho$ in the basis provided by $U_{A1}$. Since we have freedom in the choice of $U_A$ when $A$ has eigenvalues with higher multiplicity, we again need to pick it with the perturbation $E$ in mind, but now we will not use the alternative which makes $\hat E$ block diagonal. Rather, given a fixed perturbation $E$ we pick $U_A$ so that $B$ is diagonal and $D_{(i,j)}=0$ whenever $\alpha_i=\alpha_j$.

We now make yet another simplification. Since the eigenvalues of $A-\rho I_n+E$ are just translates by $\rho$ of the eigenvalues of $A+E$, there is no restriction to assume that $\rho=0$. With this assumption, it follows that $A$ has rank $m$ and that $A+E$ in the new basis can be written
\begin{equation}\label{spec_dec0}\Lambda_\alpha+\hat E=\left(
                                                                 \begin{array}{cc}
                                                                   0 & 0 \\
                                                                   0 & \Lambda_\tau \\
                                                                 \end{array}
                                                               \right)+\left(
                                                                 \begin{array}{cc}
                                                                   \Lambda_\beta+C(\Lambda_\tau+D)^{-1}C^* & C \\
                                                                   C^* & D \\
                                                                 \end{array}
                                                               \right).\end{equation}
The matrix $B$, which we denote by $\Lambda_\beta$ to stress that it is diagonal, is the so called Schur complement of the block $\Lambda_\tau+\hat E_{22}$  in the matrix $\Lambda_\alpha+\hat E$. The vector $\beta$ contains its eigenvalues, ordered non-increasingly.

We now express $\Lambda_\alpha+\hat E$ in a new basis which is not orthonormal, and thereby $\Lambda_\alpha+\hat E$ does not become self-adjoint in the new representation. The basis is given by the columns of $\left(
\begin{array}{cc}
I & 0 \\
-(\Lambda_\tau+D)^{-1}C^*  & I \\
\end{array}
\right)$. Since \begin{align*}
&\left(
\begin{array}{cc}
I & 0 \\
(\Lambda_\tau+D)^{-1}C^*  & I \\
\end{array}
\right)
\left(
\begin{array}{cc}
I & 0 \\
-(\Lambda_\tau+D)^{-1}C^*  & I \\
\end{array}
\right)=
\left(
\begin{array}{cc}
I & 0 \\
0 & I \\
\end{array}
\right).
\end{align*}
this becomes
\begin{align*}
&\left(
\begin{array}{cc}
I & 0 \\
(\Lambda_\tau+D)^{-1}C^*  & I \\
\end{array}
\right)
\left(
\begin{array}{cc}
\Lambda_\beta+C(\Lambda_\tau+D)^{-1}C^* & C \\
C^* & \Lambda_\tau+D \\
\end{array}
\right)
\left(
\begin{array}{cc}
I & 0 \\
-(\Lambda_\tau+D)^{-1}C^* & I \\
\end{array}
\right)=\\&\left(
\begin{array}{cc}
\Lambda_\beta & C \\
(\Lambda_\tau+D)^{-1}C^*\Lambda_\beta & \Lambda_\tau+D+(\Lambda_\tau+D)^{-1}C^*C\\
\end{array}
\right).
\end{align*}
In particular, the matrix
\begin{align}\label{matrix}
\left(
\begin{array}{cc}
\Lambda_\beta & C \\
(\Lambda_\tau+D)^{-1}C^*\Lambda_\beta & \Lambda_\tau+D+(\Lambda_\tau+D)^{-1}C^*C\\
\end{array}
\right)\end{align}
 shares eigenvalues with $A+E$. This matrix has elements that are $O(\|C\|)$ in the 2nd quadrant and elements that are $O(\|B\|\|C\|)$ in the 3rd quadrant, where we number the quadrants of such a matrix by $\left(
\begin{array}{cc}
1 & 2 \\
3  & 4 \\
\end{array}
\right)$. Thus, examining the determinant that leads to the characteristic polynomial, any off-diagonal contributions lead to terms which are bounded by suitable powers of $\|B\|\|C\|^2$, and this is the key observation underlying the proof of Theorem \ref{teigenvalues}, along with Cartan's theorem. We shall prove the following theorem, which is a stronger version of Theorem \ref{teigenvalues} adapted to the present environment. We omit the routine details of how to lift this result to the more general environment of Theorem \ref{teigenvalues}.
\begin{theorem}\label{teigenvalues1}
Let $A_0\in \H_{n}$ have rank $m$ and set $l=n-m$. Then there are radii $r_A,~r_E>0$ and a constant $c>0$ such that the following holds; Given any other matrix $A$ with rank $m$ and $\|A-A_0\|_2<r_A$, the eigenvalues $\xi$ of $A+E$ can be ordered such that the first $l$ satisfies \begin{equation}\label{est22}|\xi_j-\beta_j|<c\|B\|\|C\|^2,\quad 1\leq j\leq l\end{equation}
for all $E\in\H_{n}$ with $\|E\|<r_E$.
\end{theorem}

\begin{proof} First note that by \eqref{defB} and standard estimates, there exists a constant $K$ such that \begin{equation}\label{estB}\|B\|\leq K\|E\|\end{equation}
for all $E$ in a sufficiently small neighborhood of 0. By \eqref{matrix}, the eigenvalues of $A+E$ are solutions to the equation
\begin{align}\label{determinant}
\det\left(
\begin{array}{cc}
\Lambda_\beta-\zeta I& C \\
O(\|B\|\|C\|) & \Lambda_\tau-\zeta I+ O(\|E\|)  \\
\end{array}
\right)=0\end{align}
where the constants involved in the ordo terms are independent of $A$, given that $r_A$ and $r_E$ are sufficiently small (to ensure that $(\Lambda_\tau+D)^{-1}$ is uniformly bounded in $A,E$). %To see this one needs the continuity of eigenvalues, which follows e.g.~from Weyl's inequality, we omit the details.
Let $\pi_1,\ldots,\pi_{n!}$ be an enumeration of all permutations on $\{1,\ldots,n\}$ with $\pi_1=Id$. If $X(\zeta)$ denotes the matrix in \eqref{determinant}, the equation can then be written $$0=\sum_{j=1}^{n!}\mathsf{sgn}(\pi_j)\prod_{i=1}^{n}X_{i,\pi_j(i)}(\zeta)=\sum_{j=1}^{n!}\mathsf{sgn}(\pi_j)p_j(\zeta)$$ where $p_j(\zeta)=\prod_{i=1}^{n}X_{i,\pi_j(i)}(\zeta)$ and the first term is $$p_1(\zeta)=\prod_{i=1}^{n}X_{i,\pi_1(i)}(\zeta)=\prod_{i=1}^{l}(\beta_{i}-\zeta)\prod_{i=1}^{m}(\tau_{i}+O(\|E\|)-\zeta).$$
Assume that $r_A,r_E$ are chosen such that \begin{equation}\label{hf}|\tau_{i}+O(\|E\|)|\geq 3 K r_E\end{equation} for all $A$, $E$ and $1<i\leq m$. Here we use that the terms denoted $O(\|E\|)$ have uniform bounds independent of $A$. We will in the remainder often omit explicit mentioning of this detail. We remark that this is the only requirement on $r_A$, whereas the requirements on $r_E$ will be updated on several occasions. Let $\Omega_3$ (c.f. Lemma \ref{l2}) be the open disc around 0 with radius $2 K r_E$.

We will show, using induction, that there exists numbers  $c,~r_E$ and $c_1$ such that whenever $\|E\|<r_E$, each $\sum_{j=1}^{t}\prod_{i=1}^{n}X_{i,\pi_j(i)}(\zeta)$ is of the form $$\varphi_t(\zeta)\prod_{i=1}^{l}(\mu_i-\zeta)$$ (c.f. \eqref{g6}) for some numbers $\mu_1,\ldots,\mu_l$ satisfying \begin{equation}\label{t}\|\mu-\beta\|_2\leq c\|B\|\|C\|^2,\end{equation} and \begin{equation}\label{textra}\{\mu_i\}_{i=1}^l\subset \left(\sum_{j=1}^t 2^{-j}\right)\Omega_3,\end{equation} and an analytic function $\varphi_t$ satisfying \begin{equation}\label{textra1}\inf_{\zeta\in\Omega_3}|\varphi_t(\zeta)|>c_1.\end{equation}

This is certainly true for $t=1$, even with $c=0$. Indeed, then $\mu=\beta$ and the right hand side of \eqref{textra} becomes a disc with radius $K r_E$ which clearly includes $\{\beta_{i}\}_{i=1}^m$ by \eqref{estB}, so we only need to find a suitable number $c_1$ for \eqref{textra1}. Since $\varphi_1(\zeta)=\prod_{i=1}^{m}(\tau_{i}+O(\|E\|)-\zeta)$, it suffices to pick $c_1$ any number smaller than $(3Kr_E-2Kr_E)^m=(Kr_E)^m$ by \eqref{hf}.

Now suppose the induction hypothesis is true for some fixed $t\geq 1$ and consider \begin{equation}\label{pl}p_{t+1}(\zeta)=\prod_{i=1}^{n}X_{i,\pi_{t+1}(i)}(\zeta).\end{equation} Suppose that $k$ is such that $l-k$ is the amount of diagonal entries from the first quadrant, and suppose for simplicity that these are ordered such that they come on indices $i=k+1,\ldots,l$. The corresponding part of $p_{t+1}$ then equals $$\prod_{i=k+1}^{l}(\beta_{i}-\zeta).$$ Moreover, from each of the first $k$ rows we must get contributions from the 2:nd quadrant unless $p_{t+1}\equiv 0$ (since off-diagonal elements in the first quadrant are zero), i.e.~elements which are $O(\|C\|)$. Similarly, from each of the first $k$ columns we must get contributions from the 3:rd quadrant, i.e.~elements which are $O(\|B\|\|C\|)$. Returning to \eqref{pl} we thus conclude that \begin{equation}\label{pl1}p_{t+1}(\zeta)=\psi(\zeta)\prod_{i=k+1}^{l}(\beta_i-\zeta)\end{equation} where $\sup_{\zeta\in\Omega_3}|\psi(\zeta)|<c_4(\|B\|\|C\|^2)^{k}$ for some constant $c_4$ (and moreover this constant only depends on $r_A$ and $r_E$, not on $A$ or $E$).

By our induction hypothesis, $\sum_{j=1}^{t+1}\prod_{i=1}^{n}X_{i,\pi_j(i)}(\zeta)$ is of the form \begin{equation}\label{g9}\varphi_t(\zeta)\prod_{i=1}^{l}(\mu_i-\zeta)+\psi(\zeta)\prod_{i=k+1}^{l}(\beta_i-\zeta)\end{equation}
where $\inf_{\zeta\in\Omega_3}{|\varphi_t(\zeta)|}>c_1$ and $\|\mu-\beta\|_2<c\|B\|\|C\|^2.$ Upon writing
$${\psi}(\zeta)\prod_{i=k+1}^{l}(\beta_i-\zeta)={\psi}(\zeta)\prod_{i=k+1}^{l}\Big((\beta_i-\mu_i)-(\zeta-\mu_i)\Big)$$
and expanding, this term turns into a finite number of terms (at most $2^l$) which each, possibly after reordering, have the form \begin{equation}\label{gty}\tilde{\psi}(\zeta)\prod_{i=\tilde k+1}^{l}(\mu_i-\zeta)\end{equation} where $\tilde k\geq k$ and  $\sup_{\zeta\in\Omega_3}|\tilde{\psi}(\zeta)|<c_5(\|B\|\|C\|^2)^{\tilde{k}}$ for some $c_5$ (which can be chosen independent of $A$ and be valid for all such terms, here we use the induction hypothesis \eqref{t}). We now need to do another finite induction step, but for brevity we will provide less details. Adding the first such term to the first term of \eqref{g9}, i.e. to $\varphi_t(\zeta)\prod_{i=1}^{l}(\mu_i-\zeta)$, we get $$\left(\varphi_t(\zeta)\prod_{i=1}^{\tilde{k}}(\mu_i-\zeta)-\tilde{\psi}\right)\prod_{i=\tilde{k}+1}^{l}(\mu_i-\zeta).$$
By Lemma \ref{l2}, applied with $\Omega_1=\left(\sum_{j=1}^t 2^{-j}\right)\Omega_3$ and $\Omega_2=\left(\sum_{j=1}^t 2^{-j}+2^{-(t+1)}\frac{1}{2^l}\right)\Omega_3$, this can be written \begin{equation}\label{gty1}\tilde\varphi_t(\zeta)\left(\prod_{i=1}^{\tilde{k}}(\nu_i-\zeta)\right)\left(\prod_{i=\tilde{k}+1}^{l}(\mu_i-\zeta)\right) \end{equation}
where the new sequence of zeroes $(\nu_1\ldots \nu_{\tilde{k}},\mu_{\tilde{k}+1},\ldots \mu_l)$ lie in $\Omega_2$ and differs from $\mu$ by a norm bounded by the constant $c_3c_5^{1/\tilde k}$ times $\|B\|\|C\|^2$, and $\tilde\varphi_t(\zeta)$ is uniformly bounded below by $c_1/2$ on $\bar \Omega_3$.

The form of \eqref{gty} and \eqref{gty1} are the same, except that the new zeroes in \eqref{gty1} lie in the slightly larger disc $\Omega_2$. We may thus repeat the process, (albeit with new values of the involved constants), and at the $q$:th step we apply Lemma \ref{l2} with $$\Omega_1=\left(\sum_{j=1}^t 2^{-j}+2^{-(t+1)}\frac{(q-1)}{2^l}\right)\Omega_3$$ and
$\Omega_2=\left(\sum_{j=1}^t 2^{-j}+2^{-(t+1)}\frac{q }{2^l}\right)\Omega_3$. Since the process stops at most at $q=2^l$, we conclude that \eqref{g9} can be written as \begin{equation}\label{g19}\varphi_{t+1}(\zeta)\prod_{i=1}^{l}(\tilde\mu_i-\zeta)\end{equation}
with $\inf_{\zeta\in\Omega}{|\varphi_{t+1}(\zeta)|}>c_1/2^{2^l}$ and $\|\mu-\tilde\mu\|_2<c_6\|B\|\|C\|^2$ for some constant $c_6$. We also get that \eqref{textra} is satisfied and finally, by the triangle inequality we see that \eqref{t} continues to hold for the new zeroes $\tilde{\mu}$, albeit with larger constant $c+c_6$. Redefining $c_1$ to equal $c_1/2^{2^l}$ (which we can do since this is a \textit{finite} induction argument), this concludes the induction argument.

To see that this implies \eqref{est22}, note that at the end of the induction, when $t=n!$, the corresponding vector $\mu$ constitute the $l$ smallest (in modulus) solutions to \eqref{determinant}, and hence \eqref{est22} follows immediately from \eqref{t} with $\xi_j=\mu_j,$ $1\leq j \leq l$.

\end{proof}

The inversion of $\Lambda_\tau+D$, needed to compute $B$, can be computationally expensive, and therefore it is interesting to note that the theorem holds unchanged if the $D$ is removed.

\begin{corollary}
The above theorem holds with \eqref{est22} replaced by \begin{equation}\label{est22tweek}|\xi_j-\tilde\beta_j|<c\|E\|^3,\quad 1\leq j\leq l,\end{equation} where $\tilde \beta$ are the eigenvalues of $\tilde B$ defined by $\tilde B=\hat E_{11}-C\Lambda_\tau^{-1}C^*$
\end{corollary}
\begin{proof}
Clearly $B-\tilde B=C\Lambda_\tau^{-1}C^*-C(\Lambda_\tau+D)^{-1}C^*$ differ by a term which is $O(\|E\|^3)$, and thus the desired result follows by Theorem \ref{teigenvalues1} and Weyl's inequality.
\end{proof}

\section{An approximate spectral decomposition}\label{secspec}

In the previous section the matrix $U_A$ was chosen in a rather delicate manner, depending on $E$. In this section we switch the roles. Let $A=U_A\Lambda_\alpha U_A^*$ be any spectral decomposition of $A\in\H_n$, and consider perturbation of eigenvectors for $A+E$. As is well known, the eigenvectors are discontinuous. To see this, just note that $I+E$ gets the eigenvectors of $E$ which is highly discontinuous as a function of $E$ near 0. Nevertheless, given a fixed matrix $U_A$ we will distinguish a rather large subset of $E$'s where the perturbed eigenvectors are well behaved. Again, the structure of the Schur complements turns out to be decisive.

To be more specific, given a fixed eigenvalue $\rho$ of $A$, we define $B,C$ and $D$ as in the previous section, and get
\begin{equation}\label{spec_dec2}\Lambda_\alpha+\hat E=\left(
                                                                 \begin{array}{cc}
                                                                   \rho I_l & 0 \\
                                                                   0 & \Lambda_\tau \\
                                                                 \end{array}
                                                               \right)+\left(
                                                                 \begin{array}{cc}
                                                                   B+C(\Lambda_\tau-\rho I_m+D)^{-1}C^* & C \\
                                                                   C^* & D \\
                                                                 \end{array}
                                                               \right)\end{equation}
where for simplicity we assume that we have chosen the basis so that $\rho$ is the ``first'' eigenvalue, but we do not assume that $\rho=0$ as in \eqref{spec_dec0}. We will refer to $B$ simply as the Schur complement for $\rho,$ leaving its dependence on $\rho$ implicit in the notation. The eigenvalues of a given Schur complement $B$ will be denoted $\beta.$

Given any $c>0$ we define the set $\V_c$ to be the set of matrices $E$ such  that the Schur complement $B$ is diagonal, $B=\Lambda_\beta$, for each eigenvalue $\rho$ of $A$, and moreover satisfies \begin{equation}\label{conic1}|\beta_{i}-\beta_{j}|\geq c\|E\|\end{equation} whenever $i\neq j$ and $\alpha_i=\alpha_j$. Again we underline that $E$ here depends on $U_A$, in contrast to the opposite case in the previous section. For such perturbations we claim that \begin{equation}\label{vap}U_{ap}(E)=U_A(I-M\circ \hat E)\end{equation} is an approximate set of eigenvectors to $A+E$. To begin with, note that
$U_{ap}^*=(I+M\circ \hat E)U_A^*$ due to the structure of $M$, so that
$U_{ap}$ is ``orthonormal up to errors of $O(\|E\|^2)$'', in the sense that $U_{ap}^*U_{ap}= I+O(\|E\|^2)$. Using \eqref{id} we also see that
\begin{equation}\label{spec_dec_pert}A+E= U_{ap}(\Lambda_\alpha+\hat E^d)U_{ap}^*+O(\|E\|^2)\end{equation}
by which $U_{ap}$ in a sense qualifies for an approximate matrix of eigenvectors. However, the above formulas unfortunately do not imply that $U(E)=U_{ap}(E)+O(\|E\|^2)$ for any matrix of actual eigenvectors $U$. The next proposition provides this additional information, at least for $E\in\V_c$.

\begin{proposition}\label{p1}
Given $c>0$ there is a function $U:\V_c\rightarrow O_n$ of normalized eigenvectors to $A+ E$ that satisfies $$U(E)=U_{ap}(E)+O(\|\hat E\|^2),\quad E\in \mathcal{V}_{c}.$$
\end{proposition}
\begin{proof}
Let $\rho$ be an arbitrary eigenvalue of $A$. Note that the eigenvalues of $A+E$ are distinct for small enough $E$, due to Theorem \ref{teigenvalues1}, and hence the eigenvectors of $A+E$ are unique up to multiplication with unimodular numbers. Since the eigenvectors are unaffected by adding the identity matrix, there is no restriction to assume that $\rho=0$, and clearly we can assume that $A=\Lambda_\alpha$ and hence that $U_A=I_n$. We adopt all settings from the previous section. Note that $\hat E=E$ so we will omit the hat from the notation. In particular, we can write $A+E=\Lambda_\alpha+ E$ as in \eqref{spec_dec2} with $\rho=0$ and $B=\Lambda_\beta$. If $l$ is the multiplicity of $\rho$, we define $m=n-l$ and we let $q$ be any number $1\leq q\leq l.$ We let $\xi_q$ denote the eigenvalue belonging to $\rho+\beta_q$ via Theorem \ref{teigenvalues1}.

We shall show that the eigenvector for the eigenvalue $\xi_q$ has the form stipulated by the theorem. We have \begin{equation}\label{eigv}U_{ap}(E)=\left(
                                                                 \begin{array}{cc}
                                                                   I & C\Lambda_\tau^{-1} \\
                                                                   -\Lambda_\tau^{-1}C^*  & I-M_{22}\circ D \\
                                                                 \end{array}
                                                               \right)\end{equation}
(where $M_{22}$ denotes the fourth quadrant of the matrix $M$), and hence we shall show that the $q:$th column of this is at most $O(\|E\|^2)$ distance from a unit norm eigenvector to $\xi_q$. The first $l$ columns of the above matrix coincide up to errors of size $O(\|E\|^2)$ with $$\tilde U_{ap}=\left(
\begin{array}{cc}
I & 0 \\
-(\Lambda_\tau+D)^{-1}C^* & I \\
\end{array}
\right).$$ Multiplying $\Lambda_\alpha+\hat E-\xi_q I_n=\left(
                                                                 \begin{array}{cc}
                                                                   \Lambda_\beta-\xi_q I_l+C(\Lambda_\tau+D)^{-1}C^* & C \\
                                                                   C^* & \Lambda_\tau+D-\xi_q I_m \\
                                                                 \end{array}
                                                               \right)$ with this matrix we get
\begin{align}\label{matrix1}
\left(
\begin{array}{cc}
\Lambda_\beta-\xi_q I_l & C \\
\xi_q(\Lambda_\tau+D)^{-1}C^* & \Lambda_\tau+D-\xi_q I_m \\
\end{array}
\right).\end{align} Note that the desired statement follows by basic considerations if we show that \eqref{matrix1} has a concrete vector $v_q=v_q(E)$ of the form $e_q+O(\|E\|^2)$ in its kernel, where $e_1,\ldots,e_n$ is the canonical basis. Indeed, then
$$(\Lambda_\alpha+\hat E-\xi_q I_n)\tilde U_{ap} v_q=0$$ so $\xi_q$ has an eigenvector of the form $\tilde U_{ap}(e_q+O(\|E\|^2))$. Normalizing this vector amounts to division by $\sqrt{1+O(\|E\|^2)}$ which is the same as multiplying by $1+O(\|E\|^2)$, and hence this would imply that $\xi_q$ has a unit norm eigenvector of the form $\tilde U_{ap}e_q+O(\|E\|^2)$, and we can define the matrix $U=U(E)$ as the matrix whose $q:$th column is this specific eigenvector. Since $\tilde U_{ap}e_q$ and $U_{ap}e_q$ differ by a term of size $O(\|E\|^2)$, the proof would be complete.

We thus conclude by showing that \eqref{matrix1} has a vector in the kernel of the desired form. To begin with, note that \eqref{matrix1} has the form \begin{align}\label{matrix2}
\left(
\begin{array}{cc}
\Lambda_\beta-\xi_q I_l & O(\|E\|) \\
O(\|E\|^2) & \Lambda_\tau+O(\|E\|) \\
\end{array}
\right).\end{align}
Since the above kernel is one-dimensional, as noted earlier, it follows by basic linear algebra that a non-trivial vector $u$ in the kernel can be obtained as follows; let $X_p$ denote the matrix in \eqref{matrix2} with the $q:$th row replaced by $e_p$, and set $u_p=\det X_p$.

We now analyze the size of the $u_p$'s in $u$. First of all, if $p=q$ then the diagonal contribution to $u_q=\det X_q$ is bounded below (in modulus) by $c_1(\|E\|^{l-1})$ for some constant $c_1$, at least in a neighborhood of 0. Here we use the fact that $\xi_q=\beta_q+O(\|E\|^3)$ by Theorem \ref{teigenvalues1} and the assumption that  $E\in\V_c$. Any other term in the expansion of $\det X_q$ will contain factors that are at least $O(\|E\|^l)$, and hence there exists a neighborhood of 0 such that $$|u_q|\geq \frac{c_1}{2}\|E\|^{l-1}.$$

If $l < p \leq n$, the $q:$th column of $X_p$ is $O(\|E\|^2)$, whereas each of the remaining $l-1$ first columns are of size $O(\|E\|)$. Consider a non-zero term in the formula $$u_p=\det X_p=\sum_{j=1}^{n!}\mathsf{sgn}(\pi_j)\prod_{i=1}^{n}(X_p)_{i,\pi_j(i)}.$$ Since it contains precisely one element from each column, we deduce that each such term is $O(\|E\|^{l+1})$. Summing up we see that \begin{equation}\label{xr}u_p=O(\|E\|^{l+1}).\end{equation}

We finally wish to draw the same conclusion in the case when $1\leq p\leq l$ but $p\neq q$. Since the $q:$th row only contains one non-zero element (in position $p$), this one must be included in the product $\prod_{i=1}^{n}X_p({i,\pi_j(i)})$ (for it to be non-zero). This means that the $p:$th row must draw a non-diagonal element. The only non-zero such element must be in the 2nd quadrant, and hence gives a factor of $O(\|E\|)$. On top of that, there is one $O(\|E\|)$ contribution from each of the first $l$ columns, except column $p$ and also except column $q$, which actually gives an $O(\|E\|^2)$-contribution. Summing up this gives a term of $O(\|E\|\cdot\|E\|^{l-2}\cdot \|E\|^2)$, which again gives \eqref{xr}, as desired.

Now set $v=\frac{u}{u_q}$. Then $v_q=1$ and $v_p=O(\|E\|^2)$ for $p\neq q,$ and the proof is complete.
\end{proof}

Unfortunately, the set $\V_c$ is not invariant under multiplication by scalars, due to the non-linear term $C(\Lambda_\tau-\rho I_m +D)^{-1}C^*$ in the definition of the Schur complement. In particular, the above theorem can not help us to compute the Gateaux derivative. However, a given Schur complement for a given eigenvalue $\rho$ only differs from the corresponding submatrix of $E$ by the term $C(\Lambda_\tau-\rho I_m +D)^{-1}C^*$ which is $O(\|E\|^2)$. It is therefore tempting to suppose that the definition of $U(E)$ can be extended outside $\V_c$ so that the formula \eqref{vap} is the first order term also for matrices $E$ such that $\hat E_{(i,j)}=0$ and $|\hat E_{(i,i)}-\hat E_{(j,j)}|>c\|E\|$ whenever $\alpha_i=\alpha_j$ cf.~\eqref{conic1} (i.e.~for block diagonal matrices with block-wise decreasing diagonal elements). Surprisingly, this turns out to be false. We provide an example showing this, and in the next section we provide a correct formula for the Gateaux derivative.

Consider $A+E$ for the matrix $$A=\left(
                                                                                                                                           \begin{array}{ccc}
                                                                                                                                             0 & 0 & 0 \\
                                                                                                                                             0 & 0 & 0 \\
                                                                                                                                             0 & 0 & 1 \\
                                                                                                                                           \end{array}
                                                                                                                                         \right)$$
and the perturbation $E=tF$ where $$F=\left(
                                                                                                                                           \begin{array}{ccc}
                                                                                                                                             1 & 0 & 1 \\
                                                                                                                                             0 & 0 & 1 \\
                                                                                                                                             1 & 1 & 0 \\
                                                                                                                                           \end{array}
                                                                                                                                         \right)$$
and $t$ is a parameter that we will have approach zero. For the eigenvalue $\rho=0$ of $A$ we then get $$C(\Lambda_\tau+D)^{-1}C^*=\left(
                                    \begin{array}{cc}
                                      t^2 & t^2 \\
                                      t^2 & t^2 \\
                                    \end{array}
                                  \right)
$$ and $B=\left(
                                    \begin{array}{cc}
                                      t-t^2 & -t^2 \\
                                      -t^2 & -t^2 \\
                                    \end{array}
                                  \right)
$. Applying \eqref{spec_dec_pert} to  $B/t$ we conclude that $$B=t\left(
                                    \begin{array}{cc}
                                      1 & t \\
                                      -t & 1 \\
                                    \end{array}
                                  \right)\left(
                                    \begin{array}{cc}
                                      1-t & 0 \\
                                      0 & -t \\
                                    \end{array}
                                  \right)\left(
                                    \begin{array}{cc}
                                      1 & -t \\
                                      t & 1 \\
                                    \end{array}
                                  \right)+O(t^3)$$ and due to Proposition \ref{p1} we know that the approximate eigenvectors $\left(
                                    \begin{array}{cc}
                                      1 & t \\
                                      -t & 1 \\
                                    \end{array}
                                  \right)$ are correct to the first order. Let us therefore chose a new basis given by the columns of a matrix $V=V(t)$ so that the first two vectors coincide with these in the first two components. However, we also wish to pick $V$ to be unitary so we alter it as follows $V =\left(
                                                                                                                                           \begin{array}{ccc}
                                                                                                                                             \sqrt{1-t^2} & t & 0 \\
                                                                                                                                             -t & \sqrt{1-t^2} & 0 \\
                                                                                                                                             0 & 0 & 1 \\
                                                                                                                                           \end{array}
                                                                                                                                         \right).$
We then get $\tilde E=V^*EV=\left(
                                                                                                                                           \begin{array}{ccc}
                                                                                                                                             t & t^2 & t-t^2 \\
                                                                                                                                             t^2 & 0 & t+t^2 \\
                                                                                                                                             t-t^2 & t+t^2 & 0 \\
                                                                                                                                           \end{array}
                                                                                                                                         \right)+O(t^3),$ and $A+E$ gets the representation $A+\tilde E$.
Moreover, the Schur complement is now given by $$\tilde B=\left(
                                    \begin{array}{cc}
                                      t-t^2 & 0 \\
                                      0 & -t^2 \\
                                    \end{array}
                                  \right)+O(t^3)
$$
which clearly is closer to a diagonal matrix than the original Schur complement. We would like to apply Proposition \ref{p1} to $A+\tilde E$ in the new basis, but this is still impossible due to the off-diagonal $O(t^3)$-terms. However, an inspection of the proof reveals that the proposition still can be applied when if we relax the assumption that each Schur-complement $B$ is diagonal, to just assuming that the off diagonal elements are $O(\|E\|^3)$. We assume for the moment that the details of this claim has been verified, and note that when applying  Proposition \ref{p1} to $A+\tilde E$ we can take $U_A=I$, hence $\hat{\tilde{E}}=\tilde E$ and $$M=\left(
                                                                                                                                           \begin{array}{ccc}
                                                                                                                                             0 & 0 & -1 \\
                                                                                                                                             0 & 0 & -1 \\
                                                                                                                                             1 & 1 & 0 \\
                                                                                                                                           \end{array}
                                                                                                                                         \right).$$
 Proposition \ref{p1} therefore gives that  $A+\tilde E$ has normalized eigenvectors of the form $$\left(
                                                                                                                                           \begin{array}{ccc}
                                                                                                                                             1 & 0 & t \\
                                                                                                                                             0 & 1 & t \\
                                                                                                                                             -t & -t & 1 \\
                                                                                                                                           \end{array}
\right)+O(t^2),$$ which in turn gives that $A+E$ has normalized eigenvectors of the form
$$V\left(
                                                                                                                                           \begin{array}{ccc}
                                                                                                                                             1 & 0 & t \\
                                                                                                                                             0 & 1 & t \\
                                                                                                                                             -t & -t & 1 \\
                                                                                                                                           \end{array}
\right)+O(t^2)=\left(
                                                                                                                                           \begin{array}{ccc}
                                                                                                                                             1 & t & t \\
                                                                                                                                             -t & 1 & t \\
                                                                                                                                             -t & -t & 1 \\
                                                                                                                                           \end{array}
\right)+O(t^2).$$
The correct Gateaux derivative of the matrix $F$ is therefore $\left(
                                                                                                                                           \begin{array}{ccc}
                                                                                                                                             0 & 1 & 1 \\
                                                                                                                                             -1 & 0 & 1 \\
                                                                                                                                             -1 & -1 & 0 \\
                                                                                                                                           \end{array}
\right)$ and not $-M\circ F=\left(
                                                                                                                                           \begin{array}{ccc}
                                                                                                                                             0 & 0 & 1 \\
                                                                                                                                             0 & 0 & 1 \\
                                                                                                                                             -1 & -1 & 0 \\
                                                                                                                                           \end{array}
\right)$, which \eqref{vap} might falsely lead us to believe. The fact that \eqref{vap} fails is even more surprising while considering that the two leading order terms for $\xi$, as given in \eqref{tao}, are unaffected by whether $A$ has simple eigenvalues or not. Is it true that all terms in the series-expansion of $\xi(t)$ are independent in the same way?

A final remark on the example; upon evaluating the eigenvectors of $A+0.01 F$ on a computer one sees that $\left(
                                                                                                                                           \begin{array}{ccc}
                                                                                                                                             1 & 0.01 & 0.01 \\
                                                                                                                                             -0.01 & 1 & 0.01 \\
                                                                                                                                             -0.01 & -0.01 & 1 \\
                                                                                                                                           \end{array}
\right)$ are eigenvectors with 3 decimals precision. In the next section we prove the above claims in a general framework.

\section{Perturbations along a line and Gateaux differentiability}\label{secline}

In this section we consider perturbations of the form $A+tF$ for general $A,F\in\H_n$, where $t$ is a real parameter. We pick a unitary matrix $U_{A,F}$ that diagonalizes $A$ and moreover is chosen so that $\hat F=U_{A,F}^*FU_{A,F}$ satisfies $\hat F_{(i,j)}=0$ whenever $\alpha_i=\alpha_j$, i.e.~such that $\hat F$ is block-wise diagonal. Given a diagonal matrix $\Lambda$ we recall the notation $\Lambda^{\dagger}$ for the Moore-Penrose inverse, which in this case reduces to another diagonal matrix whose diagonal elements are the inverse of those in $\Lambda$, except when the diagonal element is 0 in which case $\Lambda^\dagger$ has a 0 as well. Theorem \ref{teigenvalues1} can in this situation be made more explicit, as follows.

\begin{theorem}\label{t4}
Given matrices $A,F\in \H_n$ and $U_{A,F}\in O_n$ such that $F$ is blockwise diagonal, the eigenvalues  $\xi(t)$ of $A+tF$ can be ordered so that they are real analytic functions with Taylor expansion at 0 given by \begin{equation}\label{r}\xi_j(t)=\alpha_j+t\hat F_{(j,j)}-t^2(\hat F^*(\Lambda_\alpha-\alpha_j I)^\dagger \hat F)_{(j,j)}+O(t^3).\end{equation}
\end{theorem}

Note that $\xi$ is not necessarily ordered decreasingly in the above theorem, and that the formula \eqref{tao} from the introduction is a simple reformulation of \eqref{r}. The real analyticity of $\xi$ is a classical result by F. Rellich, see e.g. Theorem 6.1 in Chapter II of \cite{kato2013perturbation}, the main feature of the above theorem is formula \eqref{r} which gives concrete expressions for the first 3 terms, and this will be shown along with Theorem \ref{t5} below. As mentioned in the introduction, the formula \eqref{r} appears in \cite{lancaster1964eigenvalues} and can also be found in older texts on quantum mechanics, such as \cite{hilbert1955methods}, but this is rather hard to see and in either case the proofs presented here are new.

Theorem 6.1 by \cite{kato2013perturbation} also provides real analyticity of the corresponding eigenprojections, and it is even shown that it is possible to chose the eigenvectors as real analytic functions as well (as functions of $t$). This is a most remarkable result due to the fact that it is impossible to even define the eigenvectors continuously (as a function of $E$). In any case, a concrete formula for the leading order terms is very hard to find, although it is implicit in the details of Chapter 5.13 of \cite{hilbert1955methods}. We repeat this find here with a new proof.

 Assuming that $\hat F$ has block-wise decreasing diagonal elements, we recall that $\hat F$ is unique and introduce the matrix $N=N(\hat F)$ as follows;
\begin{equation}\label{N}N_{(i,j)}=\left\{\begin{array}{ll}
      0, & \text{if $\alpha_i\neq \alpha_j$ or $i=j$}  \\
      \frac{(F^*(\Lambda_\alpha-\alpha_jI_n)^\dagger F)_{(i,j)}}{\hat F_{(i,i)}-\hat F_{(j,j)}}, & \text{else}
    \end{array}\right.
\end{equation}
In the example from the previous section we would get $N(\hat F)=\left(
                                                                                                                                           \begin{array}{ccc}
                                                                                                                                             0 & 1 & 0 \\
                                                                                                                                             -1 & 0 & 0 \\
                                                                                                                                             0 & 0 & 0 \\
                                                                                                                                           \end{array}
\right)$ (where implicitly $U_{A,F}=I$ so $\hat F=F$). The key result of this section reads as follows:

\begin{theorem}\label{t5}
Given matrices $A,F\in \H_n$ and $U_{A,F}\in O_n$ such that $\hat F$ is block-wise diagonal with block-wise decreasing diagonal elements, there exists a unitary matrix $U(t)$ defined for $t$ in a neighborhood of 0, such that \begin{itemize}\item[$a)$] $U$ is real analytic and the $j$:th column is an eigenvector for the eigenvalue $\xi_j$ to $A+tF$ via \eqref{r}. \item[$b)$] $U(t)=U_{A,F}+tU_{A,F}(N(\hat F)-M\circ\hat F)+O(t^2)$. \end{itemize}
\end{theorem}
In particular, the Gateaux type derivative $U'(0)=\lim_{t\rightarrow 0}{(U(t)-U(0))}/t$ in the direction $F$ exists and equals $U'(0)=U_{A,F}(N(\hat F)-M\circ\hat F).$ However, this is not a Gateaux derivative in the strict sense since $U(E)$ is not a well defined function near $E=0$, (in case $A$ has non-simple eigenvalues).
\begin{proof}
We simply write $N$ for $N(\hat F)$ and introduce the matrix $V_{ap}(t)=I+tN$. Moreover we let $V(t)$ be the unitary matrix obtained from $V_{ap}(t)$ by performing a Gram-Schmidt orthonormalization. Note that the scalar product of any two columns in $V_{ap}$ belonging to different ``blocks'' (defined by indices $(i,j)$ such that $\alpha_i=\alpha_j$) is 0, and that the scalar product of two columns belonging to the same block is $O(t^2)$, due to the particular structure of $N$. It follows that $V(t)$ is supported only within the diagonal blocks, and moreover that \begin{equation}\label{o0}V(t)=V_{ap}(t)+O(t^2).\end{equation}

We will now prove formula \eqref{r} as well as the affirmations $a)-b)$. For this it clearly suffices to focus on a fixed $j$, and as in previous sections we can assume that $\alpha_j$ equals 0 and is the ``first eigenvalue''. Clearly we may assume that $U_{A,F}=I$. Correspondingly $\hat F=F$ and hence we omit the ``hats'' from the notation in the following. We decompose $A+tF=\Lambda_\alpha+tF$ as
\begin{equation}\label{spec_dec3}\Lambda_\alpha+tF=\left(
                                                                 \begin{array}{cc}
                                                                   0 & 0 \\
                                                                   0 & \Lambda_\tau \\
                                                                 \end{array}
                                                               \right)+t\left(
                                                                 \begin{array}{cc}
                                                                   \Lambda_\varphi & Y \\
                                                                   Y^* & Z \\
                                                                 \end{array}
                                                               \right)\end{equation} where $\varphi_1,\ldots,\varphi_l$ are the diagonal values of $F$ in the first quadrant, which we also denote by $F_{11}$ in accordance with  \eqref{u7}. Recall that $F_{11}$ is assumed to be a diagonal matrix, hence $\Lambda_{\varphi}$ is an $l\times l$-matrix as opposed to $\Lambda_\alpha$ which is $n\times n$. To keep notation separate, we reserve $F_{(1,1)}$ (with a parenthesis and comma) for the element on position $(1,1)$ of $F$.
The first quadrant of the matrix $N=N(F)$ then becomes
$$(N_{11})_{(i,j)}=\left\{\begin{array}{ll}
\frac{(Y\Lambda_\tau^{-1}Y^*)_{(i,j)}}{\varphi_i-\varphi_j},& 1\leq i\neq j\leq l \\
0,& 1\leq i= j\leq l
                      \end{array}\right..
$$  Also note that $N^*=-N$, $N_{12}=N_{21}=0$ and $$\Lambda_\varphi N_{11}- N_{11}\Lambda_\varphi=(Y\Lambda_\tau^{-1}Y^*)-(Y\Lambda_\tau^{-1}Y^*)^d$$ due to its special structure.
We now express $\Lambda_\alpha+tF$ in the new basis provided by $V=V(t)$, clearly $V^*\Lambda_\alpha V=\Lambda_\alpha$ and $V^*F V$ equals
\begin{align*}&\left(\left(
                                                                 \begin{array}{cc}
                                                                   I_l-tN_{11} & 0 \\
                                                                   0 & I_m-tN_{22} \\
                                                                 \end{array}
                                                               \right)+O(t^2)\right)\left(
                                                                 \begin{array}{cc}
                                                                   \Lambda_\varphi & Y \\
                                                                   Y^* & Z \\
                                                                 \end{array}
                                                               \right)\left(\left(
                                                                 \begin{array}{cc}
                                                                   I_l+tN_{11} & 0 \\
                                                                   0 & I_m+tN_{22} \\
                                                                 \end{array}
                                                               \right)+O(t^2)\right)\\&=\left(
                                                                 \begin{array}{cc}
                                                                   \Lambda_\varphi+t((Y\Lambda_\tau^{-1}Y^*)-(Y\Lambda_\tau^{-1}Y^*)^d)+O(t^2) & Y+O(t) \\
                                                                   Y^*+O(t) & Z+O(t) \\
                                                                 \end{array}
                                                               \right)\end{align*}
$\Lambda_\alpha+tF$ is thus unitarily equivalent with \begin{equation}\label{gt6}\left(
                                                                 \begin{array}{cc}
                                                                   0 & 0 \\
                                                                   0 & \Lambda_\tau \\
                                                                 \end{array}
                                                               \right)+\left(
                                                                 \begin{array}{cc}
                                                                   t\Lambda_\varphi+t^2((Y\Lambda_\tau^{-1}Y^*)-(Y\Lambda_\tau^{-1}Y^*)^d)+O(t^3) & tY+O(t^2) \\
                                                                   tY^*+O(t^2) & tZ+O(t^2) \\
                                                                 \end{array}
                                                               \right),\end{equation} and in particular they share eigenvalues.
Now define $E=E(t)=tV(t)^*FV(t)$ and then define $B,C,D$ as in the decomposition \eqref{spec_dec2} (with $\rho=0$). Note that also $B,C$ and $D$ become functions of $t$. Clearly then $C(t)=tY+O(t^2)$, $D(t)=tZ+O(t^2)$ whereas $B(t)$ gets a more complicated expression. In fact,
\begin{align*}&B(t)=\Big(t\Lambda_\varphi+t^2((Y\Lambda_\tau^{-1}Y^*)-(Y\Lambda_\tau^{-1}Y^*)^d)+O(t^3)\Big)-C(t)(\Lambda_\tau+D(t))^{-1}C^*(t)=\\
&\Big(t\Lambda_\varphi+t^2((Y\Lambda_\tau^{-1}Y^*)-(Y\Lambda_\tau^{-1}Y^*)^d)+O(t^3)\Big)-\Big(t^2Y\Lambda_\tau^{-1}Y^*+O(t^3)\Big)=\\&t\Lambda_\varphi-t^2(Y\Lambda_\tau^{-1}Y^*)^d+O(t^3).\end{align*}
Note in particular that in the original notation, the values of $B(t)$ on the diagonal are $t F_{(j,j)}-t^2(F^*(\Lambda_\alpha-\alpha_j)^{\dagger}F)_{(j,j)}+O(t^3)$, whereas all off-diagonal entries are $O(t^3)$. To prove the identity \eqref{r}, we now have to return to the proof of Theorem \ref{teigenvalues1}. If we replace any occurrence of $\|B\|,\|C\|,\|D\|,\|E\|$ by $t$, it is easy to check that the majority of the proof goes through with minimal changes and that \eqref{r} follows by \eqref{est22}. An exception is the paragraph following \eqref{pl}, which needs to be modified due to the off-diagonal elements in $B(t)$. However, since these are $O(t^3)$, the argument is easily updated to conclude that \eqref{pl1} holds with $\sup_{\zeta\in\Omega_3}|\psi(\zeta)|<c_4(t^3)^{k}$ for some constant $c_4$, and from there on the proof is easy. We omit the details.

The argument does not show that the eigenvalues are real analytic, but this follows from F. Rellich's theorem mentioned earlier. We need to verify that the ordering of the eigenvalues in this theorem coincides with the ordering used in Rellich's theorem, but this is easy since subsets of $\R^2$ of the form $$\{(t,\alpha_j+t\hat F_{(j,j)}+(\hat F^*(\Lambda_\alpha-\alpha_j)^\dagger \hat F)_{(j,j)}t^2+\epsilon: |t|<\delta, |\epsilon|\leq ct^3\}$$ only intersect at (0,0), given that $\delta$ is small enough. Here we again use that $\hat F$ has block-wise distinct diagonal elements.

We now consider the statements concerning the eigenvectors. If we apply the procedure described in the proof of Proposition \ref{p1} to the matrix \eqref{gt6}, we get in place of \eqref{matrix2} a matrix with the similar structure \begin{align*}
\left(
\begin{array}{cc}
t\Lambda_\varphi-t^2(Y\Lambda_\tau^{-1}Y^*)^d-\xi_q I_l+O(t^3) & O(t) \\
O(t^2) & \Lambda_\tau+O(t) \\
\end{array}
\right)\end{align*}
where $\xi_q=t\varphi_q-t^2(Y\Lambda_\tau^{-1}Y^*)_{q,q}+O(t^3)$ by the first part of the proof. With suitable modification of the proof of Proposition \ref{p1}, it follows that a matrix $U(t)$ containing a set of normalized eigenvectors of the form $I-M\circ E+O(t^2)$ exists. By \eqref{gt6} we see that $E=tF+O(t^2)$ so the normalized eigenvectors have the form $I-t M\circ F+O(t^2)$. Combining this with \eqref{o0} and the definition of $V_{ap}$, we see that the original matrix $\Lambda_\alpha+tF$ has normalized eigenvectors of the form $$(I+tN+O(t^2))(I-t M\circ F+O(t^2))=I+tN-t M\circ F+O(t^2),$$ as was to be shown. By inspection of the construction it also follows that $U(t)$ is real analytic, (given real analyticity of $\xi$).
\end{proof}

\bibliographystyle{plain}
\bibliography{MCref}

\end{document}